\documentclass{article}

\usepackage[margin=1in]{geometry}
\usepackage{amsmath,amssymb,amsthm}
\usepackage[hidelinks]{hyperref}
\usepackage{theoremref}
\usepackage{mathtools}
\usepackage{algorithm}
\usepackage{algpseudocode}
\usepackage{subcaption}
\usepackage{tikz}
\usepackage{multirow}
\usepackage{booktabs}
\usepackage{adjustbox}
\usepackage{physics}
\usepackage{float}
\usepackage{enumitem}

\newtheorem{theorem}{Theorem}
\newtheorem{lemma}{Lemma}

\newtheorem{proposition}{Proposition}

\usetikzlibrary{arrows}
\usetikzlibrary{shapes.geometric}
\usetikzlibrary{positioning}
\usetikzlibrary{intersections, calc, arrows}

\makeatletter
\newcommand{\figcaption}[1]{\def\@captype{figure}\caption{#1}}
\newcommand{\tblcaption}[1]{\def\@captype{table}\caption{#1}}
\makeatother

\newcommand{\bF}{\mathbb{F}}
\newcommand{\cB}{\mathcal{B}}
\newcommand{\cF}{\mathcal{F}}
\newcommand{\cI}{\mathcal{I}}
\newcommand{\cS}{\mathcal{S}}
\newcommand{\cU}{\mathcal{U}}
\newcommand{\cZ}{\mathcal{Z}}

\let\oldComment\Comment
\renewcommand{\Comment}[1]{\oldComment{\textcolor{gray}{#1}}}

\title{A Faster Deterministic Algorithm for Mader's $\mathcal{S}$-Path Packing}

\author{
Satoru Iwata\thanks{
    Department of Mathematical Informatics, 
    The University of Tokyo, Tokyo 113-8656, Japan, and  
    Institute for Chemical Reaction Design and Discovery, Hokkaido University, Sapporo, Hokkaido 001-0021, Japan. 
    \href{mailto:iwata@mist.i.u-tokyo.ac.jp}{iwata@mist.i.u-tokyo.ac.jp}
  }
  \and 
  Hirota Kinoshita\thanks{
    Department of Mathematical Informatics, 
    The University of Tokyo, Tokyo 113-8656, Japan. 
    \href{mailto:hirotak@g.ecc.u-tokyo.ac.jp}{hirotak@g.ecc.u-tokyo.ac.jp}
  }
}

\date{}

\begin{document}

\maketitle

\begin{abstract}
  Given an undirected graph $G = (V, E)$ with a set of terminals $T\subseteq V$ partitioned into a family $\cS$ of disjoint blocks, find the maximum number of vertex-disjoint paths whose endpoints belong to two distinct blocks while no other internal vertex is a terminal.
  This problem is called Mader's $\cS$-path packing. It has been of remarkable interest as a common generalization of the non-bipartite matching and vertex-disjoint $s\text{-}t$ paths problem. 
  
  This paper presents a new deterministic algorithm for this problem via known reduction to linear matroid parity. 
  The algorithm utilizes the augmenting-path algorithm of Gabow and Stallmann (1986), while replacing costly matrix operations between augmentation steps with a faster algorithm that exploits the original $\cS$-path packing instance. 
  The proposed algorithm runs in $O(mnk)$ time, where $n = |V|$, $m = |E|$, and $k = |T|\le n$. This improves on the previous best bound $O(mn^{\omega})$ for deterministic algorithms, where $\omega\ge2$ denotes the matrix multiplication exponent.
\end{abstract}

\section{Introduction}\label{sec:intro}

Let $G = (V, E)$ be an undirected graph of $n = |V|$ vertices and $m = |E|$ edges.
We are given a set $T\subseteq V$ of $k = |T|$ \textit{terminals} as well as a partition $\cS$ of $T$ into non-empty \textit{blocks}.
A path on $G$ is said to be a $\cS$-path if and only if each endpoint is a terminal in a distinct block from the other and no internal vertex is a terminal.
A $\cS$-path packing is a set of vertex-disjoint $\cS$-paths.
Mader's $\cS$-path packing problem is to find a $\cS$-path packing of the maximum possible size.

This problem was raised by Gallai~\cite{gallai1961maximum}, generalizing both the non-bipartite matching problem (when $\cS = \{\{v\}\mid v\in V\}$) and the vertex-disjoint $R$-$S$ paths problem (when $\cS = \{R,S\}$ for mutually disjoint terminal sets $R,S\subseteq V$). 
For an equivalent form of the openly disjoint $T$-paths problem, Mader~\cite{mader1978die} established a good characterization by proving a min-max theorem. Lov{\'a}sz~\cite{lovasz1978matroid,lovasz1980matroid} then showed a reduction to matroid matching and developed the first polynomial-time algorithm. A linear representation of this reduction is described in Schrijver~\cite{schrijver2003combinatorial}, allowing for more efficient algorithms via linear matroid parity. Schrijver~\cite{schrijver2001short} also provided a short inductive proof of the min-max theorem of Mader. 

The fastest known deterministic algorithm to Mader's problem entirely relies on reduction to linear matroid parity as well.
Despite the reduced instance having a special structure,
the running time of this algorithm is $O(mn^{\omega})$ directly implied by
Gabow and Stallmann's \textit{augmenting-path} algorithm~\cite{gabow1986augmenting},
which has been the fastest deterministic solution for general linear matroid parity problems.
Note that $\omega\ge 2$ denotes the matrix multiplication exponent with the current best bound of $\omega < 2.371552$ due to \cite{williams2024new}.

We overcome this shortcoming and present a new deterministic algorithm for Mader's problem. 
Our approach circumvents matrix inversion and multiplication used to create and update an auxiliary graph in the Gabow--Stallmann algorithm~\cite{gabow1986augmenting}. Instead, we update the auxiliary graph using combinatorial techniques, which exploit an interpretation of the parity solution as a subgraph of $G$ in the original instance of Mader's problem. The proposed algorithm thus involves only simple arithmetic operations over a finite field $\bF_q$ of order $q=O(|\cS|)$.  Throughout this paper, we adopt the uniform-cost model of computation, i.e., any single arithmetic operation over $\bF_q$ is performed in constant time.
\begin{theorem}\thlabel{thm:main}
  There exists a deterministic algorithm that solves Mader's $\cS$-path packing problem in $O(mnk)$ time.
\end{theorem}

\subsection{Related work}

\paragraph{Randomized algorithms} 
Extending the algebraic algorithms for bipartite matching by  Sankowski~\cite{sankowski2006weighted}, Harvey~\cite{harvey2009algebraic} provided efficient randomized algebraic algorithms for matching and linear matroid intersection.  Subsequently, Cheung, Lau, and Leung~\cite{cheung2014algebraic} extended this approach to the linear matroid parity problem. A direct application of the resulting algorithm to Mader's $\cS$-paths packing problem runs in $O(mn^{\omega-1})$ time. Furthermore, they improved their randomized algorithm to run in $O(n^\omega)$ time, faster than ours for dense graphs.

\paragraph{Combinatorial algorithms}
Subsequently to the min-max theorem, considerable efforts has been made to obtain efficient combinatorial algorithms.
For the non-zero $A$-path packing problem on group-labeled graphs, which generalizes Mader's $\cS$-paths packing problem, Chudnovsky, Cunningham, and Geelen~\cite{chudnovsky2008algorithm} provided an $O(n^5)$-time combinatorial algorithm. This algorithm yields the Edmonds--Gallai type structural decomposition independently of the prior work of Seb{\H{o}} and Szeg{\H{o}}~\cite{sebHo2004path}.
Furthermore, Pap~\cite{pap2008packing} devised a polynomial-time combinatorial algorithm based on a different type of graph decomposition, which generalizes the odd ear-decomposition of factor-critical graphs~\cite{lovasz1972note}.

\paragraph{Mader matroid}

It is known that Mader's $\mathcal{S}$-path problem has a more intrinsic matroid interpretation, known as \textit{Mader matroid}~\cite{schrijver2003combinatorial}. 
The Mader matroid is a matroid $(T, \mathcal{I})$ on the ground set $T$ where each basis is the set of endpoints of a distinct maximum $\cS$-path packing, that is,
\begin{align}
  \mathcal{I} = \{T'\subseteq T \mid \mbox{There exists a $\cS$-path packing that covers $T'$.}\}.
\end{align}
Schrijver's question of whether every Mader matroid is a gammoid~\cite{schrijver2001each, schrijver2003combinatorial} has been positively resolved by Pap~\cite{pap2006mader}, which hence implies a linear representation of the Mader matroid~\cite{mason1972class}.
However, the representation requires as large a field as $2^{|T|}$, 
and it is still unknown if there exists a deterministic algorithm to obtain the representation in polynomial time. 

\paragraph{Mader delta-matroid}

A recent work~\cite{wahlstrom2024representative} establishes the linear delta-matroid structure in Mader's $\mathcal{S}$-path problem, called the \textit{Mader delta-matroid}, in parallel with the relation between the matching matroid and the matching delta-matroid.
Specifically, the Mader delta-matroid is defined as the delta-matroid $(T,\cF)$ with
\begin{align}
  \cF = \{T' \subseteq T \mid \mbox{The exists a $\cS$-path packing whose endpoints are $T'$.} \}.
\end{align}
This interpretation also allows them to show that any instance $(G,T,\cS)$ can be converted to a \textit{mimicking} instance $(\tilde G,T,\cS)$ with $|V(\tilde G)| = O(|T|^3)$.
However, they rely on randomized algorithms to compute the skew-symmetric matrix representing the Mader delta-matroid, involving proper instantiation of variables in the Tutte matrix.
Efficient algorithms to find the maximum feasible set of a linear delta-matroid either require such a representation explicitly~\cite{koana2024faster} or use the \textit{separation oracle}, which is also computationally hard for the Mader delta-matroid, in a greedy approach~\cite{bouchet1987greedy,chandrasekaran1988pseudomatroids}.

\paragraph{Variants/Generalizations of Mader's problem}

Prior to Mader's theorem, Gallai~\cite{gallai1961maximum} showed the min-max theorem for the special case of $\cS = \{\{t\}\mid t\in T\}$ known as $T$-path packing, which contributes to the $O(k^2)$-vertex kernel for feedback vertex set~\cite{thomasse2010kernel}.
The edge-disjoint companion of Mader's problem is solved by a combinatorial algorithm~\cite{iwata2023finding} in $O(m^2)$ time, faster than handling line graphs in Mader's problem.
Mader's problem and its variants has been generalized to weighted graphs in many forms~\cite{karzanov1993edge,karzanov1997multiflows,pap2011polynomial,hirai2014tree}, among which Yamaguchi~\cite{yamaguchi2016shortest} established the first polynomial-time algorithm for the minimum weight $\cS$-path packing via reduction to weighted linear matroid parity~\cite{iwata2021weighted,pap2013weighted}.
Other extended notions include non-returning $A$-paths~\cite{pap2007packing,pap2008packing} as well as non-zero $A$-paths in group-labelled graphs~\cite{chudnovsky2006packing,chudnovsky2008algorithm,tanigawa2016packing,yamaguchi2016packing}.
Half-integral, or more generally, node-capacitated $\cS$-path packing has also been ddressed~\cite{pap2007some,pap2008strongly,babenko2009fast,hirai2018dual}, playing a key role in many FPT algorithms~\cite{iwata2018csps}.
\section{Reduction to linear matroid parity}\label{sec:reduce}

It suffices to consider the case when the graph $G$ is connected, as the general graph can be handled in each connected component;
this would not increase the total running time from what \autoref{thm:main} guarantees for the original instance.
We also assume $|\cS|\ge 2$; otherwise the problem would be trivial.
As is common in the literature, we denote
$E[X, Y]\coloneqq \{\{x,y\}\in E \mid x\in X,y\in Y\}$ for any $X,Y\subseteq V$. In particular, for each $X\subseteq V$, 
we denote $E[X]\coloneqq E[X, X]$ and let $G[X] = (X, E[X])$ denote the subgraph induced by $X$.

We use the following simplified form of Schrijver's linear representation~\cite{schrijver2003combinatorial}. A similar formulation is used in~\cite{cheung2014algebraic}.
Let $q$ be the least prime number that satisfies $|\cS|<q\leq 2\,|\cS|$. Note that $q$ is odd as $|\cS|\geq 2$. Let $\bF_q$ be the finite field of order $q$ so that each terminal $t\in S\in\cS$ can be labelled with a scalar $\theta(t)\in\bF_q$ uniquely to its block $S$ i.e.,
\begin{align}
  \theta(t) = \theta(t') \quad&\Leftrightarrow\quad (\{t,t'\}\subseteq S,\ \exists S\in\cS), 
  &\forall t,t'\in T.
\end{align}
Each edge $e\in E$ is also equipped with an arbitrary bijection $\mu_e:e\to\{1,-1\}\subseteq\bF_q$ that orients $e$, although this orientation has no effect on the output or any combinatorial operation of our algorithm.
We then introduce a $(2n-k)$-dimensional vector space $W$ over $\bF_q$ spanned by a \textit{singleton} $\vb*{t}$ for each $t\in T$ as well as \textit{vertex-twins} $\vb*{v^\circ}$, $\vb*{v^\bullet}$ for each $v\in V\setminus T$.
Then, let
\begin{align}
  K \coloneqq \{\vb*{t} \mid t\in T\}
\end{align}
denote the set of all singletons, and let also vertex-twins for each terminal $t\in T$ be defined as
\begin{align}
  \vb*{t^{\circ}} \coloneqq \vb*{t},\quad
  \vb*{t^{\bullet}} \coloneqq \theta(v)\vb*{t}.\label{eq:def-terminal-twin}
\end{align}
Let $\ell_e \coloneqq \{\vb*{e}^\circ, \vb*{e}^\bullet\}$ be a \textit{line} that consists of \textit{edge-twins}:
\begin{align}
  \vb*{e}^\circ &\coloneqq \mu_e(u) \vb*{u^{\circ}} + \mu_e(v) \vb*{v^{\circ}},\label{eq:def-edge-circ}\\
  \vb*{e}^\bullet &\coloneqq \mu_e(u) \vb*{u^{\bullet}} + \mu_e(v) \vb*{v^{\bullet}}\label{eq:def-edge-bullet}
\end{align}
for each edge $e\in E$, and let
\begin{align}
  L \coloneqq \bigcup_{e\in E} \ell_e = \{\vb*{e}^\circ,\vb*{e}^\bullet \mid e \in E\}
\end{align}
denote the union of all lines.
Let also
\begin{align}
  J \coloneqq L \cup \{\vb*{v^{\circ}},\vb*{v^{\bullet}} \mid v\in V\}
\end{align}
contain all vertex-twins and edge-twins.
A vector set is said to be \textit{feasible} if and only if it is a union of some lines plus some singletons.
Then let $\cU\subseteq 2^{W}$ denote the collection of all feasible sets i.e.,
\begin{align}
  \cU \coloneqq \left\{\{\vb*{t} \mid t\in S\} \cup \{\vb*{e}^\circ, \vb*{e}^\bullet \mid e\in F\} \mid S\subseteq T, F\subseteq E \right\},
\end{align}
and let $\cI\subseteq\cU$ (resp. $\cB\subseteq\cU$) be the collection of all feasible independent sets (resp. all feasible bases of $W$).
For any feasible set $U\in\cU$, we define
\begin{align}
  T[U] &\coloneqq \{t\in T \mid \vb*{t}\in U \},\\
  E[U] &\coloneqq \{e\in E\mid \ell_{e} = \{\vb*{e}^\circ,\vb*{e}^\bullet\}\subseteq U\},\\
  G[U] &\coloneqq (V, E[U]).
\end{align}
as the subset of terminals, the subset of edges, and the subgraph \textit{induced} by $U$, respectively.
Let $\cZ(U)$ denote the set of all connected components in $G[U]$.
More precisely, $\cZ(U)$ is a partition of $V$, where each $Z\in\cZ(U)$ is a maximal set of vertices that induces a connected subgraph $G[Z]$.
\autoref{lem:walk}, which immediately follows by definition, leads to \autoref{lem:base-exist} ensuring the existence of a feasible base.

\begin{lemma}\thlabel{lem:walk}
  If there is a walk on the graph $G$ that passes vertices $v_0, v_1 \ldots, v_k\in V$ and edges $e_1,\ldots,e_k\in E$ in this order, we have
  \begin{align}
    \sum_{i=1}^k \mu_{e_i}(v_i)\, \vb*{e_i^\circ} &= \vb*{v_k^\circ} - \vb*{v_0^\circ},\label{eq:path-circ}\\
    \sum_{i=1}^k \mu_{e_i}(v_i)\, \vb*{e_i^\bullet} &= \vb*{v_k^\bullet} - \vb*{v_0^\bullet}.\label{eq:path-bullet}
  \end{align}
  Particularly, the following also holds if $v_0,v_k\in T$:
  \begin{align}
    \theta(v_k)\sum_{i=1}^k \mu_{e_i}(v_i) \vb*{e}_i^\circ - \sum_{i=1}^k \mu_{e_i}(v_i) \vb*{e}_i^\bullet
    = (\theta(v_k)-\theta(v_0))\vb*{v_0}.\label{eq:path-terminal}
  \end{align}
\end{lemma}

\begin{lemma}\thlabel{lem:base-exist}
  There exists a feasible base $B\in\cB$ that contains all singletons i.e., $K\subseteq B$.
\end{lemma}

\begin{proof}
  As $G$ is connected, there is a spanning forest $(V, F)$ of $G$ consisting of $|T|$ connected components (trees) each of which contains a distinct terminal in $T$.
  Then, the feasible set
  \begin{align}
    B \coloneqq K \cup \bigcup_{f\in F} \ell_f = K \cup \{\vb*{f}^\circ,\vb*{f}^\bullet \mid f\in F\}
  \end{align}
  is a base of $W$ due to~\autoref{lem:walk}.
\end{proof}

We are ready to consider a linear matroid parity problem in the form of finding a feasible base $B\in\cB$ that maximizes $|E[B]|$, where it holds by definition that
\begin{align}
  |T[B]| + 2 |E[B]| &= |B| = \dim\, W =  2n - k, 
  &\forall B\in\cB.
  \label{eq:base-size}
\end{align}
Note that \autoref{lem:base-exist} makes this equivalent to seeking the maximum number of lines, or the maximum \textit{matching} $M$ such that $\{\vb*{e}^\circ,\vb*{e}^\bullet \mid e\in M\}$ is independent.
To see the validity of this reduction, we establish the following characterization of feasible independent sets and of feasible bases.

\begin{lemma}\thlabel{lem:indep-char}
  A feasible set $U\in\cU$ is independent i.e., $U\in\cI$ if and only if every $Z\in\cZ(U)$ satisfies all the following conditions:
  \begin{enumerate}[label={\rm (\alph*)}]
    \item $(G[U])[Z]$ is a tree.\label{item:tree}
    \item $|Z\cap S| \le 1,\quad\forall S\in\cS$.\label{item:one-per-block}
    \item $|Z\cap T| + |Z \cap T[U]| \le 2$.\label{item:terminal-le}
  \end{enumerate}
\end{lemma}

\begin{proof}
  For any $U\in\cU$ and $Z\in\mathcal{Z}(U)$, let
  \begin{align}
    U[Z] &\coloneqq \{\vb*{t} \mid t \in Z\cap T[U] \} \cup \bigcup_{e\in E[Z] \cap E[U]}\ell_e\\
    &= \{\vb*{t} \mid t \in Z\cap T[U] \} \cup \{\vb*{e}^\circ, \vb*{e}^\bullet \mid e\in E[Z] \cap E[U]\}.
  \end{align}
  By the definition \eqref{eq:def-terminal-twin}--\eqref{eq:def-edge-bullet}, $U\in\cU$ is independent if and only if $U[Z]$ is independent for every $Z\in\mathcal{Z}(U)$.

  First, we prove the necessity in the lemma.
  Fix any feasible independent set $U\in\cI$ and $Z\in\mathcal{Z}[U]$.
  Eq.~\eqref{eq:path-circ}--\eqref{eq:path-bullet} in \autoref{lem:walk} requires that no cycle exists in $(G[U])[Z]$ for the independence of $U$, which ensures~\ref{item:tree};
  Eq.~\eqref{eq:path-terminal}
  necessitates~\ref{item:one-per-block} for $U\in\cI$ as well.
  It also follows from the definition \eqref{eq:def-terminal-twin}--\eqref{eq:def-edge-bullet} that
  \begin{align}
    |U[Z]| = \rank\, U[Z]
    &\le |Z\cap T| + 2|Z\setminus T|\\
    &= 2(|Z| - 1) - |Z\cap T| + 2\\
    &= |U[Z]| - |Z\cap T[U]| - |Z\cap T| + 2,
  \end{align}
  which yields~\ref{item:terminal-le} as desired.

  Next, we prove the sufficiency in the lemma.
  Fix an arbitrary independent set $U\in\cU$ that satisfies \ref{item:tree}--\ref{item:terminal-le}.
  As initially noted, it suffices to show that $U[Z]\in\mathcal{I}$ for every $Z\in\mathcal{Z}(U)$;
  given \ref{item:terminal-le}, we separately deal with three cases depending on $|Z\cap T|\in\{0,1,2\}$.
  When $Z\cap T = \emptyset$, it is easy from \ref{item:tree} and the definition \eqref{eq:def-terminal-twin}--\eqref{eq:def-edge-bullet} to see $U[Z]\in\mathcal{I}$.
  If $Z\cap T = \{t\}$ for some $t$, then \ref{item:terminal-le} implies either $Z\cap T[U] = \emptyset$ or $Z\cap T[U] = \{t\}$.
  It hence follows from \ref{item:tree}--\ref{item:terminal-le} and Eq.~\eqref{eq:path-circ}--\eqref{eq:path-bullet} that
  \begin{align}
    \rank (U\cup\{\vb*{t}\})[Z] = 1 + 2(|Z| - 1) = |(U\cup\{\vb*{t}\})[Z]|,\label{eq:proof-indep-one-terminal}
  \end{align}
  which requires $U[Z]\in\cI$.
  Finally, when $|Z\cap T| = \{t_1, t_2\}$ for some $t_1\neq t_2$, we have $Z\cap T[U] = \emptyset$ from \ref{item:terminal-le}.
  Here, Eq.~\eqref{eq:path-circ}--\eqref{eq:path-terminal} together imply that two singletons $\vb*{t}_1,\vb*{t}_2$ as well as twins $\vb*{v}^\circ,\vb*{v}^\bullet$ for every $v\in Z\setminus T$ can be all represented by $U$.
  This guarantees that
  \begin{align}
    \rank U[Z] = 2 + 2(|Z| - 2) = 2(|Z| - 1) = |U[Z]|,
  \end{align}
  or equivalently $U[Z]\in\cI$ as desired.
\end{proof}

\begin{lemma}\thlabel{lem:base-char}
  A feasible independent set $I\in\cI$ is a base of $W$ i.e., $I\in\cB$ if and only if every $Z\in\cZ(U)$ satisfies all the following conditions:
  \begin{enumerate}[label={\rm (\alph*)},start=4]
    \item $Z\cap T\neq\emptyset$,\label{item:non-empty}
    \item $|Z\cap T| + |Z \cap T[I]| = 2$.\label{item:terminal-eq}
  \end{enumerate}
\end{lemma}

\begin{proof}
  First, we prove the necessity in the lemma.
  Fix any feasible base $I\in\cB$ and $Z\in\cZ(I)$.
  If $Z\cap T = \emptyset$, there exists some edge $e\in E[Z, V\setminus Z]$ since $G$ is connected;
  then despite $I\cup\ell_e\in\cI$, we would have $e\notin E[I]$ i.e., $\ell_{e}\not\subseteq I$ by definition of $Z$, which contradicts $I\in\cB$.
  Hence, \ref{item:non-empty} is necessary.
  It remains for \ref{item:terminal-le} to show that $|Z\cap T[I]| \ge 1$ holds when $|Z\cap T| = 1$, given the necessity of \ref{item:terminal-le} and \ref{item:non-empty}.
  In this case, Eq.~\eqref{eq:proof-indep-one-terminal} follows from the proof of \autoref{lem:indep-char} with $U$ replaced by $I$, which requires $\vb*{t}\in I$ for $I\in\cB$ as desired.

  Next, we prove the sufficiency in the lemma.
  Fix any feasible independent set $I\in\cI$ for which
  \ref{item:non-empty}-\ref{item:terminal-eq} holds for every $Z\in\cZ(I)$.
  It then follows that
  \begin{align}
    |I| = |T[I]| + 2|E[I]|
    &= \sum_{Z\in\mathcal{Z}(I)} (|Z\cap T[I]| + 2|E[Z] \cap E[I]|)\\
    &= \sum_{Z\in\mathcal{Z}(I)} (2 - |Z\cap T| + 2(|Z| - 1))\\
    &= \sum_{Z\in\mathcal{Z}(I)} (2|Z| - |Z\cap T|)
    = 2|V| - |T|
    = \dim W,
  \end{align}
  or equivalently $I\in\cB$ as desired.
\end{proof}

\begin{proposition}[\cite{schrijver2003combinatorial}]\thlabel{prop:pack-base}
  There exists a $\cS$-path packing of size $p$ if and only if there exists a feasible base $B\in\cB$ that contains $n - k + p$ lines i.e., $|E[B]| = n - k + p$.
\end{proposition}

\begin{proof}
  First, we prove the necessity in the proposition.
  Let $P_1,\ldots,P_p$ be arbitrary $p$ vertex-disjoint $\cS$-paths such that $s_i, t_i\in T$ be the endpoints of each path $P_i$ i.e.,
  \begin{align*}
    \{t_1^{(i)}, t_2^{(i)}\} &\coloneqq P_i\cap T,&\forall i\in\{1,\ldots,p\}.
  \end{align*}
  Notice $2p\le |T| = k$ as these paths are vertex-disjoint.
  Let also
  \begin{align*}
    \{ t_1^{(0)}, \ldots, t_{k - 2p}^{(0)} \} &\coloneqq T\setminus \{t_1^{(i)}, t_2^{(i)} \mid i\in\{1,\ldots,p\} \}
  \end{align*}
  denote the set of terminals disjoint from the $\cS$-paths.
  We can construct a spanning forest $H = (V, F)$ that satisfies all the following conditions:
  \begin{itemize}
    \item The vertex set $V$ is partitioned into $k - p$ connected components $X_1,\ldots,X_p,Y_1,\ldots,Y_{k-2p}$.
    \item For each $i\in\{1,\ldots,p\}$, the subgraph $H[X_i]$ contains the $\cS$-path $P_i$ and satisfies $X_i\cap T = \{t_1^{(i)}, t_2^{(i)}\}$.
    \item For each $j\in\{1,\ldots,k-2p\}$, we have $Y_j\cap T = \{t_j^{(0)}\}$.
  \end{itemize}
  Now, we define a feasible set $B\in\cU$ as follows:
  \begin{align}
    B \coloneqq \{ \vb*{t}_1^{(0)}, \ldots, \vb*{t}_{k - 2p}^{(0)} \} \cup \bigcup_{f\in F} \ell_f
    = \{ \vb*{t}_1^{(0)}, \ldots, \vb*{t}_{k - 2p}^{(0)} \} \cup \{ \vb*{f}^\circ, \vb*{f}^\bullet \mid f\in F \}.
  \end{align}
  \autoref{lem:indep-char} and \autoref{lem:base-char} ensures $B\in\cB$, while it is clear that
  \begin{align}
    |E[B]| = |F| = n - k + p.
  \end{align}

  Next, we prove the sufficiency in the proposition.
  Let $B\in\cB$ be an arbitrary feasible base that consists of $n-k+p$ lines plus $k-2p$ singletons.
  Then, Lemmas \ref{lem:indep-char} and \ref{lem:base-char} imply that $G[B]$ partitions $V$ into connected $k-p$ components $X_1,\ldots,X_p,Y_1,\ldots,Y_{k-2p}$ that satisfies the following:
  \begin{itemize}
    \item For each $i\in\{1,\ldots,p\}$, we have $|X_i\cap T| = 2$ as well as $|X_i \cap S| \le 1,\ \forall S\in\cS$.
    \item For each $j\in\{1,\ldots,k-2p\}$, we have $|Y_j\cap T| = 1$.
  \end{itemize}
  Each connected subgraph $(G[B])[X_i]$ has a $\cS$-path, which constitutes a desired $\cS$-path packing of size $p$.
\end{proof}

\section{Revisiting the Gabow--Stallmann algorithm}\label{sec:revisit}

Before presenting our algorithm, we review how  Gabow and Stallmann~\cite{gabow1986augmenting} would solve the linear matroid parity problem.
Their augmenting-path method generalizes similar concepts for matroid intersection~\cite{krogdahl1977dependence,lawler1976combinatorial} and that for non-bipartite matching~\cite{gabow1976efficient,edmonds1965paths,even1975n2}.
It starts from an arbitrary feasible base $B\in\cB$, and then repeatedly updates the feasible base $B$ with $E[B]$ increased by $2$ (and hence $T[B]$ decreased by $1$) per iteration.
Specifically, a new feasible base $B'$ results from taking the symmetric difference of the original base $B$ and an \textit{augmenting path} $\vb*{s},\ell_0,\ell_1,\ldots,\ell_{2p},\vb*{t}$ consisting of an odd number of lines with two singletons, as follows:
\begin{align}
  B' \setminus B &= \bigcup_{i=0}^p \ell_{2i},\\
  B \setminus B' &= \{\vb*{s},\vb*{t}\}\cup \bigcup_{i=1}^p \ell_{2i - 1}.
\end{align}
Such an appropriate path is sought on a bipartite graph between $B$ and $L\setminus B$, called a \textit{dependence graph}~\cite{gabow1986augmenting}, which connects $\vb*{b}\in B$ and $\vb*{w}\in L\setminus B$ with a scalar $d(\vb*{w}, \vb*{b})$ if and only if $\vb*{b}$ has a non-zero coefficient $d(\vb*{w}, \vb*{b})$ in the representation of $\vb*{w}$ with respect to $B$.
In other words, augmenting a given feasible base $B\in\cB$ requires us to compute the (unique) matrix $D(B) = (d(\vb*{w},\vb*{b}))_{\vb*{w}\in L\setminus B,\,\vb*{b}\in B}\in\bF_q^{(L\setminus B)\times B}$ that enjoys
\begin{align}
  \vb*{w} &= \sum_{\vb*{b}\in B} d(\vb*{w}, \vb*{b}) \vb*{b},
  &\forall\vb*{w}\in L\setminus B,
  \label{eq:def-depmat}
\end{align}
which we refer to as the \textit{dependence matrix} for the feasible base $B\in\cB$ in this paper.

Here, it must be carefully considered that their original setting differs from~\autoref{sec:reduce} in the definition of singletons and hence in the underlying vector space;
they define singletons as copies of all single twins in $L$.
These singletons can be padded to each matching to form a base of $\mathrm{span}\,L$, not of $W$.
However, the singletons do not serve as anything but dummy vectors to ensure the existence of as large an independent set as the number of lines, which obviously bounds the size of any matching.
Indeed, replacing these singletons still preserves every guarantee in their proof as long as this property holds, which implies \autoref{prop:aug}.

\begin{proposition}[\cite{gabow1986augmenting}]\thlabel{prop:aug}
  There exists a deterministic algorithm that, given any feasible base $B$ coupled with the dependence matrix $D(B)$, returns either a new feasible base $B'$ enjoying $|E[B']| = |E[B]| + 1$, or reports if $E[B]$ is the maximum possible, in $O(mn)$ time and $O(mn)$ space under the uniform cost model.
\end{proposition}

As a whole, the bottleneck of their algorithm lies in repeatedly computing the dependence matrix $D(B)$ each time $B$ is updated by augmentation.
By the above definition~\eqref{eq:def-depmat} of $D(B)$, this generally requires $O(m n^{\omega-1})$ time~\cite{aho1974design}, which dominates the $O(mn)$ time for a single augmentation guaranteed in~\autoref{prop:aug}.
Since $\rank\,L\le n$, the time complexity amounts to $O(mn^{\omega})$ in total~\cite{gabow1986augmenting}.

\begin{algorithm}[tb]
  \caption{The main algorithm for Mader's $\cS$-path packing problem.}
  
  \begin{algorithmic}[1]
    \State Find the field $\bF_q$ s.t. $|\cS| < q \le 2|\cS|$.
    \State $B \gets$ \Call{InitializeBase}{\null}
    \Comment{Initialize the base.}
    \Repeat
    \State $B' \gets B$
    \State $D \gets$ \Call{ComputeDependence}{$B$}
    \Comment{Compute the dependence matrix.}
    \State $B \gets$ \Call{AugmentOrMaximum}{$B,D$}
    \Comment{Run the Gabow-Stallmann algorithm.}
    \Until{$B\neq B'$}
    \Comment{Continue if the base has been updated.}
    \State \Return \Call{RestorePacking}{$B$}
    \Comment{Recover a $\cS$-path packing from the base.}
  \end{algorithmic}
  
  \label{alg:main}
\end{algorithm}

\begin{algorithm}[tb]
  \caption{The subroutine to initialize a feasible base.}

  \begin{algorithmic}[1]
    \Function{Initialize}{\null}
    \State $B \gets K$
    \Comment{Include all singletons.}
    \State $R \gets T$

    \While{$R\neq V$}
    \Comment{Obtain the desired spanning forest.}
    \State Pick any $e\in E[R, V\setminus R]$.
    \State $B \gets B \cup \ell_{e} = B\cup \{ \vb*{e}^\circ,\vb*{e}^\bullet \}$
    \State $R \gets R\cup e$
    \EndWhile

    \State \Return $B$
    \Comment{Returns the feasible base claimed in \autoref{lem:base-exist}.}
    \EndFunction
  \end{algorithmic}
  
  \label{alg:init}
\end{algorithm}

\begin{algorithm}[tb]
  \caption{The subroutine to compute the dependence matrix for a feasible base.}
  
  \begin{algorithmic}[1]
    \Function{ComputeDependence}{$B$}
    \For{$v\in V$}
    \State $\phi(v) \gets \emptyset$
    \EndFor

    \For{$t\in T$}
    \Comment{Choose each terminal $t$ as the root.}
    \If{$\vb*{t} \in B$}
    \State $c(\vb*{t^\circ}, \vb*{t}) \gets 1$
    \State $c(\vb*{t^\bullet}, \vb*{t}) \gets \theta(t)$
    \ElsIf{$\phi(t)\neq \emptyset$}
    \Comment{Obtain the representation of ${t^\circ},{t^\bullet}$ from the tentative representations.}
    \For{$\vb*{b}\in B$}
    \State $c(\vb*{t^\circ}, \vb*{b}) \gets \cfrac{\phi(t) c(\vb*{t^\circ}, \vb*{b}) - c(\vb*{t^\bullet}, \vb*{b})}{\phi(t) - \theta(t)}$
    \State $c(\vb*{t^\bullet}, \vb*{b}) \gets \theta(t) c(\vb*{t^\circ}, \vb*{b})$
    \EndFor
    \EndIf
    \Comment{If neither applies,
      this iteration computes the tentative representations.
    }

    \State $Q \gets \{t\}$

    \While{$Q\neq\emptyset$}
    \Comment{Traverse the connected component (a tree) that contains $t$.}
    \State Remove any $u \in Q$.
    \State $\phi(u) \gets \theta(t)$

    \For{$e = \{u,v\}\in E[B]$ s.t. $\phi(v)\neq \theta(t)$}
    \State $Q \gets Q\cup \{v\}$

    \For{$\vb*{b}\in B\setminus\ell_{e}$}
    \State $c(\vb*{v^\circ}, \vb*{b}) \gets c(\vb*{u^\circ}, \vb*{b})$
    \State $c(\vb*{v^\bullet}, \vb*{b}) \gets c(\vb*{u^\bullet}, \vb*{b})$
    \EndFor

    \State $c(\vb*{v^\circ}, \vb*{e}^\circ) \gets c(\vb*{u^\circ}, \vb*{e}^\circ) + \mu_e(v)$
    \State $c(\vb*{v^\bullet}, \vb*{e}^\bullet) \gets c(\vb*{u^\bullet}, \vb*{e}^\bullet) + \mu_e(v)$

    \EndFor
    \EndWhile
    \EndFor

    \For{$(e = \{u, v\}, \vb*{b})\in (E\setminus E[B])\times B$}
    \Comment{Compute the dependence matrix.}
    \State $d(\vb*{e}^\circ,\vb*{b}) \gets \mu_e(u) c(\vb*{u^\circ},\vb*{b}) + \mu_e(v) c(\vb*{v^\circ},\vb*{b})$
    \State $d(\vb*{e}^\bullet,\vb*{b}) \gets \mu_e(u) c(\vb*{u^\bullet},\vb*{b}) + \mu_e(v) c(\vb*{v^\bullet},\vb*{b})$
    \EndFor

    \State \Return $D = (d(\vb*{w},\vb*{b}))_{\vb*{w}\in L\setminus B,\,\vb*{b}\in B}$
    \EndFunction
  \end{algorithmic}
  
  \label{alg:depmat}
\end{algorithm}

\begin{algorithm}[tb]
  \caption{The subroutine to restore a $\cS$-path packing  from a feasible base.}

  \begin{algorithmic}[1]
    \Function{RestorePacking}{$B$}
    \State $\mathcal{P} \gets \emptyset$

    \For{$v\in V$}
    \State $\psi(v) \gets \emptyset$
    \EndFor

    \For{$t\in T\setminus T[B]$}
    \If{$\psi(t) = \emptyset$}\Comment{Traverse the connected component (a tree) that contains $t$.}
    \State $Q \gets \{t\}$

    \State $\psi(t) \gets t$

    \While{$Q\neq\emptyset$}
    \State Remove any $w \in Q$.
    \For{$e = \{v, w\}\in E[B]$ s.t. $\psi(v)=\emptyset$}
    \State $Q \gets Q\cup \{v\}$
    \State $\psi(v) \gets w$
    \EndFor
    \EndWhile
    \Else\Comment{Track a $\cS$-path using the back pointers.}
    \State $v_0 \gets t$
    \State $w \gets t$
    \State $k \gets 0$

    \While{$\psi(w)\neq w$}
    \State $w \gets \psi(w)$
    \State $k \gets k + 1$
    \State $v_k \gets w$
    \EndWhile

    \State $\mathcal{P} \gets \mathcal{P} \cup \{ (v_0, \ldots, v_k) \}$
    \EndIf
    \EndFor

    \State \Return $\mathcal{P}$
    \EndFunction
  \end{algorithmic}
  
  \label{alg:restore}
\end{algorithm}

\section{Our algorithm}\label{sec:algo}

By utilizing the modified definition of singletons in \autoref{sec:reduce}, we achieve the following two points that help us improve upon the time complexity described in \autoref{sec:revisit}.
First, \autoref{lem:base-exist} allows us to start from a ``nice'' base $B\in\cB$ that accepts at most $\left\lfloor\frac{|T|}{2}\right\rfloor$ augmentations until it becomes optimal;
this bound may not hold in Gabow and Stallmann's setting.
Second, the feasible base $B\in\cB$ maintained by \autoref{alg:main} now has a helpful interpretation as shown in \autoref{lem:indep-char} and \autoref{lem:base-char}.
These are fundamental in our key subroutine to more efficiently compute the dependence matrix $D(B)$ as detailed below.

We present \autoref{alg:main} as our main routine, which is shown to correctly solve Mader's problem and achieve \autoref{thm:main}.
It is fundamentally similar to the flow of the Gabow--Stallmann algorithm discussed in \autoref{sec:revisit}, where the number of lines $|E[B]|$ in a feasible base $B\in\cB$ is increased by one per iteration until $|E[B]|$ becomes as large as possible.
The following discussions in this section establish \autoref{thm:correct}.

\begin{proposition}\thlabel{thm:correct}
  The main algorithm {\em (\autoref{alg:main})} correctly solves Mader's $\cS$-path packing problem for an arbitrary given instance $(G, T, \cS)$.
\end{proposition}

%

\subsection{Initializing the feasible base}

The subroutine \textproc{InitializeBase} defined in \autoref{alg:init} obtains the initial feasible base $B\in\cB$ constructed in the proof of \autoref{lem:base-exist}.
We can compute the desired spanning forest by expanding  disjoint vertex sets each of which contains exactly one terminal.

\subsection{Computing the dependence matrix}\label{sec:algo:depmat}

The subroutine \textproc{ComputeDependence} obtains the dependence matrix $D(B)$ for a given feasible base $B$ as defined in \autoref{alg:depmat}.
Here, it clearly suffices to compute the (unique) matrix $C = (c(\vb*{w},\vb*{b}))_{\vb*{w}\in J,\,\vb*{b}\in B}\in\bF_q^{J\times B}$ that enjoys
\begin{align}
  \vb*{w} &= \sum_{\vb*{b}\in B} c(\vb*{w}, \vb*{b})\vb*{b}
  &\forall\vb*{w}\in J,
  \label{eq:def-auxmat}
\end{align}
as the representation of each vector in $L\setminus B$ is obtained from those of at most two vectors in $J$ by definition.

Starting from each terminal chosen as the \textit{root}, this algorithm traverses the forest $G[B]$ while computing the representation of each singleton and vertex-twin with respect to $B$.
Each connected component $Z$ of $G[B]$ is scanned exactly $|Z\cap T| = 2 - |Z\cap T[B]| \in\{1,2\}$ times due to \autoref{lem:base-char}.
The former case, when $Z\cap T = \{t\}$ for some $t$, can be handled simply;
since $Z\cap T[B] = \{t\}$ and thus $\vb*{t}\in B$ hold then, we can compute the representation of twins $v^\circ, v^\bullet$ for each $v\in Z$ during the traversal from the root $t$.

In the latter case when $Z\cap T[B] = \emptyset$ and $Z\cap T = \{t_1, t_2\}$ for some $t_1\neq t_2$, assume $t_1$ is chosen as the root before $t_2$ without loss of generality.
We compute a \textit{tentative} representation of $v^\circ - t_1, v^\bullet - t_1$ for each $v\in Z$ during the first traversal from the root $t_1$.
Before the second traversal from the root $t_2$, the correct representation of the singleton $\vb*{t}_2$ can be now obtained using the tentative representations and Eq.~\eqref{eq:def-terminal-twin}, which finally allows us to propagate the correct representations across $Z$.

\subsection{Augmenting the feasible base}

The subroutine \textproc{AugmentOrMaximum} designates an oracle that takes in the current feasible base $B\in\cB$ coupled with the dependence matrix $D(B)$, and executes the augmenting-path algorithm claimed by \autoref{prop:aug}.
We assume that, when $B$ allows for no augmentation as it has already reached the maximum $E[B]$, \textproc{AugmentOrMaximum} just returns the original feasible base $B$;
otherwise, an updated feasible base is returned, and the loop continues in \autoref{alg:main}.

\subsection{Restoring a $\cS$-path packing}

After all iterations, the subroutine \textproc{RestorePacking} defined in \autoref{alg:restore} recovers the maximum $\cS$-path packing in the same way as we prove the sufficiency in \autoref{prop:pack-base}.
Recall the decomposition we use in the proof;
due to \autoref{lem:indep-char} and \autoref{lem:base-char}, the spanning forest $G[B]$ has $E[B] - n + k$ connected components that contains exactly two terminals from $T\setminus T[B]$ and from distinct blocks, while any other component contains exactly one from $T[B]$.
\autoref{alg:restore} traverses each of the former components from one of the two terminals while creating back pointers, which are then used to construct a $\cS$-path from the other terminal.

\section{Complexity}

In this section, we analyze the time and space complexity of \autoref{alg:main} to prove \autoref{prop:main-time} below. 
We assume the uniform model where any single arithmetic operation over the field $\bF_q$ of size $O(|\cS|)$ is done in constant time.
Recall that $n = |V|$, $m = |E| \ge n - 1$ as $G$ is connected, and $|\cS| \le k = |T| \le n = |V|$.

\begin{lemma}\label{lem:main}
  The main routine in {\rm \autoref{alg:main}} can find the finite field $\bF_q$ in $O(|\cS|^2)$ time.
\end{lemma}

\begin{proof}
  One can naively find the minimum prime number greater than $|\cS|$ in $O(|\cS|^2)$ time.  
\end{proof}

\begin{lemma}
  The subroutine \textproc{InitializeBase} in {\rm \autoref{alg:init}} runs in $O(n)$ time and $O(n)$ space.
\end{lemma}

\begin{proof}
  Let $B\in\cB$ be the feasible base finally obtained.
  Each vertex or edge of $G[B]$ is scanned at most once in the traversal.
  One can also avoid scanning any edge in $E\setminus E[B]$ by efficient implementation.
\end{proof}

\begin{lemma}
  The subroutine \textproc{ComputeDependence} in {\rm \autoref{alg:depmat}} runs in $O(mn)$ time and $O(mn)$ space.
\end{lemma}

\begin{proof}
  As described in \autoref{sec:algo:depmat}, each vertex or edge of $G[B]$ is scanned at most twice.
  On visiting each vertex, the representations of its corresponding vertex-twins are coordinate-wise computed in $O(n)$ time.
  Lastly, we compute $O(mn)$ entries of the dependence matrix, which dominates in the total complexity.
\end{proof}

\begin{lemma}
  The loop in {\rm \autoref{alg:main}} iterates at most $\left\lfloor\frac{k}{2}\right\rfloor$ times.
\end{lemma}

\begin{proof}
  The claim follows from \autoref{prop:aug} and Eq.~\eqref{eq:base-size}.
\end{proof}

\begin{lemma}\label{lem:restore}
  The subroutine \textproc{RestorePacking} in {\rm \autoref{alg:restore}} runs in $O(n)$ time and $O(n)$ space.
\end{lemma}

\begin{proof}
  Each vertex or edge in $G[B]$ is scanned at most once in the traversal.
  Constructing the $\cS$-paths takes $O(n)$ time and $O(n)$ space in total as those paths are vertex-disjoint.
\end{proof}

\begin{proposition}\thlabel{prop:main-time}
  The main routine {\rm \autoref{alg:main}} runs in $O(mnk)$ time and $O(mn)$ space.
\end{proposition}

\begin{proof}
  The claim follows from Lemmas \ref{lem:main}--\ref{lem:restore} as well as \autoref{prop:aug}.
\end{proof}

Finally, Propositions \ref{thm:correct} and \ref{prop:main-time} together derive \autoref{thm:main}.

\bibliographystyle{plain}
\bibliography{references}

\end{document}